\documentclass[a4paper]{article}
\usepackage{latexsym}
\usepackage{amsmath,amssymb,amsfonts}
\usepackage{verbatim}
\usepackage{hyperref}
\usepackage{graphicx} %eps figures can be used instead
\renewcommand{\baselinestretch}{1.2}
\long\def\/*#1*/{}
\newtheorem{prethm}{{\bf Theorem}}
\newenvironment{thm}{\begin{prethm}{\hspace{-0.5
               em}{\bf.}}}{\end{prethm}}

\newtheorem{prepro}[prethm]{Proposition}

\newtheorem{prelem}[prethm]{Lemma}
\newenvironment{lem}{\begin{prelem}{\hspace{-0.5
               em}{\bf.}}}{\end{prelem}}

\newtheorem{precor}[prethm]{Corollary}
\newenvironment{cor}{\begin{precor}{\hspace{-0.5
               em}{\bf.}}}{\end{precor}}

\newtheorem{preremark}{{\bf Remark}}

\newenvironment{rem}{\begin{preremark}\em{\hspace{-0.5
              em}{\bf.}}}{\end{preremark}}

\newtheorem{preexample}{{\bf Example}}

\newtheorem{preque}{{\bf Question}}

\newtheorem{preproblem}{{\bf problem}}

\newtheorem{preproof}{{\bf Proof.}}

\newenvironment{proof}[1]{\begin{preproof}{\rm
               #1}\hfill{$\Box$}}{\end{preproof}}

\renewcommand{\thefootnote}

\begin{document}
\title{On the Decision Number of Graphs}
\author{S. Akbari$^{\,\rm a,c}$,~ M. Dalirrooyfard$^{\,\rm a,b}$,~ S. Davodpoor$^{\,\rm c}$,~ K. Ehsani$^{\,\rm b}$,~ R. Sherkati$^{\,\rm b}$\\
{\footnotesize {\em $^{\rm a}$Department of Mathematical Sciences, Sharif University of Technology,}}\\
{\footnotesize {\em P.O. Box 11155-9415,
Tehran, Iran}}\\
{\footnotesize {\em $^{\rm b}$Department of Computer Engineering, Sharif University of Technology,}}\\
{\footnotesize {\em P.O. Box 11155-9517,
Tehran, Iran}}\\
{\footnotesize {\em $^{\rm c}$School of Mathematics, Institute for Research in Fundamental Sciences (IPM),}}\\
{\footnotesize {\em P.O. Box 19395-5746, Tehran, Iran}}}
\footnotetext{E-mail Addresses: {\tt s\_akbari@sharif.edu}, {\tt mdalir\_rf@ee.sharif.edu}, {\tt sodavodpoor@gmail.com}, {\tt kehsani@ce.sharif.edu}, {\tt sherkati@ce.sharif.edu}}
\date{}
\maketitle

%%%%%%%%%%%%%%%%%%%%%%%%%%ABSTRACT

\begin{abstract}

 Let $G$ be a graph. A good function is a function $f:V(G)\rightarrow \{-1,1\}$, satisfying $f(N(v))\geq 1$, for each $v\in V(G)$, where $ N(v)=\{u\in V(G)\, |\, uv\in E(G) \} $ and $f(S) = \sum_{u\in S} f(u)$ for every $S \subseteq V(G) $. For every cubic graph $G$ of order $ n, $ we prove that $ \gamma(G) \leq \frac{5n}{7} $ and show that this inequality is sharp. A function $f:V(G)\rightarrow \{-1,1\}$ is called a nice function, if $f(N[v])\le1$, for each $v\in V(G)$, where $ N[v]=\{v\} \cup N(v) $. Define $\overline{\beta}(G)=max\{f(V(G))\}$, where $f$ is a nice function for $G$. We show that $\overline\beta(G)\ge -\frac{3n}{7}$ for every cubic graph $G$ of order $n$, which improves the best known bound $-\frac{n}{2}$.
\vspace{1mm} {\renewcommand{\baselinestretch}{1}
\parskip = 0 mm

}

%\vspace{-3mm}\hfill{\rule{13.3cm}{.1mm}\hskip2cm} 

\end{abstract}

%%%%%%%%%%%%%%%%%%%%%%%%%%INTRODUCTION

\section{Introduction}
All graphs considered in this paper are finite and simple. Let $ G $ be a graph
with the vertex set $V(G)$ and the edge set $E(G)$. We denote $|V(G)|$ and $|E(G)|$ the order and size of $G$, respectively. For $v\in V(G)$, the neighborhood and the closed neighborhood of $v$ are defined by $ N(v)=\{u\in V(G)\, | \, uv\in E(G) \} $ and $ N[v]=\{v\} \cup N(v) $, respectively. Denote $d_G(v)=|N(v)|$, for simplicity we use $d(v)$ instead of $d_G(v)$. For a graph $ G $, let $ \delta(G) $ and $\Delta(G)$ denote the minimum and the maximum degree of $G$, respectively.\\
In this paper, $P_n, C_n$ and $K_n$ denote the path, the cycle and the complete graph of order $n$, respectively. Also, $K_{m,n}$ denotes the complete bipartite graph with partition $U$ and $V$, where $|U| = m$ and $|V| = n$. Moreover, $d_G(u,v)$ is the distance between $u$ and $v$. For abbreviation, sometimes we use $d(u,v)$.
%A graph $G$ with vertex set $V$, is called bipartite, if $V$ can be partitioned into sets $V_1$ and $V_2$, so that every edge of $G$ joins a vertex of $V_1$ and $V_2$. The $\it{complete bipartite graph}\rm, k_{m,n}$ is the bipartite graph with $V_1

For a function $f:V(G)\rightarrow \{-1,1\}$ and $S\subseteq V(G)$, we define $f(S) = \sum_{u\in S} f(u)$. A function $f:V(G)\rightarrow \{-1,1\}$ is called a $\emph{bad function}$, if $f(N(v))\le 1$ for each $v\in V(G)$. The maximum value of $f(V(G))$, taken over all bad function $f$, is called the {\it bad decision number} of $G$, and denoted by $\beta(G)$. The bad decision number was introduced by Changping Wang in \cite{neg} as the $\emph{negative decision number}$. If the closed neighborhood is used in the above definition, $f$ is called a $\emph{nice function}$. The {\it nice decision number} of $G$, denoted by $\overline{\beta}(G)$, is the maximum value of $f(V(G))$, taken over all nice function $f$. A function $f:V(G)\rightarrow \{-1,1\}$ is called a $\emph{good function}$, if $f(N(v))\ge 1$, for each $v\in V(G)$. The minimum value of $f(V(G))$, taken over all good function $f$, is called the {\it good decision number} of $G$, and denoted by $\gamma(G)$. If the closed neighborhood is used in the above definition, $f$ is called an $\emph{excellent function}$. The {\it excellent decision number} of $G$, denoted by $\overline{\gamma}(G)$, is the minimum value of $f(V(G))$, taken over all nice function $f$.

A \it{rooted tree }\rm $T$ distinguishes one vertex $r$ which is called \it{root}\rm. For each vertex $r\neq v \in V(T)$, the \it{parent }\rm of $v$ is the neighbor of $v$ on the unique $r–v$ path, while a \it{child }\rm of $v$ is any other neighbor of $v$. A \it{vertex coloring }\rm of $G$ is a function $f:V(G)\longrightarrow C$, where $C$ is a set of colors. A vertex coloring is called {\it proper} if adjacent vertices of $G$ receive distinct colors. The \it{total dominating set }\rm $S$ is a set of vertices such that each $v\in V(G)$ (even those in $S$) has at least a neighbor in $S$. Let $\gamma_t(G)=min{|S|}$, where $S$ takes over all total dominating sets. 

In \cite{sub}, it was proved that for every tree $T$, $\beta(T) \geq 0$. Also in \cite{neg} Wang proved that if $G$ is a graph of order $n$ and $\delta(G)\geq 2$, then $\beta(G)\leq n+1- \sqrt{ 4n + 1}$, and this bound is sharp. Moreover, he showed that $ \beta(G)\leq \frac{1}{5}(4m-3n) $, for every graph $G$ of order $n$ and size $m$, where $\delta(G) \geq 2$. Wang also proved that if $G$ is a $k$-regular graph of order $n$, then 
$$
\beta(G) \leq  
\left\{
\begin{array}{lr}
0 & $if $ k$ is even$ \\
\frac{n}{k}& $ if $k$ is odd$,\\
\end{array}
\right.
$$
 and this upper bound is sharp. Furthermore, he showed that, for any positive integer $n\geq 3$,\\
\begin{center}
$
\beta(C_{n}) =  
\left\{
\begin{array}{lr}
0 &$ if $ \, n \, \equiv \, 0\,(mod\, 4) \\
-2&$ if $ \, n \, \equiv \, 2 \,(mod\, 4)\\
-1& $otherwise,$\\
\end{array}
\right.
$\,\,\,\,\,\,\,\,
 and\,\,\,\,\,\,\,\, 
$
\beta(P_{n}) =  
\left\{
\begin{array}{lr}
0 &$ if $ \, n \, \equiv \, 0\,(mod\, 4) \\
2& $if $ \, n \, \equiv \, 2 \,(mod\, 4)\\
1& $otherwise.$\\
\end{array}
\right.
$\\
\end{center}

In this paper, we prove that $\beta(T)\leq n - 2\lceil \frac{n+6}{10} \rceil $, and $\overline{\gamma}(T) \ge n - 2\lceil\frac{n-4}{3}\rceil $, where $T$ is a tree of order $n\geq 3$. Also we show that $\overline{ \beta}(T)\geq 0 $, for every tree $ T $. It is shown that for every cubic graph $G$ of order $ n $, $ \beta(G) \ge 0$, $ {\overline{\beta}}(G) \geq \frac{-3n}{7}$. We also prove that $\gamma(G)\leq\frac{5n}{7}$ for every cubic graph $G$ of order $n$. In this paper by a short proof we show that for every cubic graph $G$ of order $n$, except Petersen graph, $\overline{\gamma}(G)\le \frac{3n}{4}$.

%%%%%%%%%%%%%%%%%%%%%%%%%%%%%%%%%%%%   BAD DECISION NUMBER

\section{Bad Decision Number}
For every positive integer $n$, $a_n$ denotes $\max \beta(T)$, taken over all tree $T\in T_n$, where $T_n$ is the set of all trees of order $n$. In this section, we show that $a_n=n-2\lceil\frac{n+6}{10}\rceil$, for every positive integer $n \geq 3$, and we characterize all trees attain this value. We present a lower bound for bad decision number of a graph in terms of its size and order. Finally, we prove that for every cubic graph $ G $, $\beta(G)\ge 0$ and show that this lower bound is sharp.

\begin{lem}
\label{lemstep}
{
For every positive integer $n$, $\mathopen| a_{n+1} - a_n \mathclose| \le 1$.
}
\end{lem}

\begin{proof}
{It is clear that $a_n \equiv n \, \mbox{(mod 2)}$. Obviously, $a_{n+1}-a_n\ge -1$. Suppose that $T_{n+1}^\ast$ is a tree such that $\beta(T_{n+1}^\ast)=a_{n+1}$. Let $f$ be a bad function for $T^\ast_{n+1}$, where $f(V(T^\ast_{n+1})) = a_{n+1}$. Let $u$ be a pendant vertex in $T_{n+1}^\ast$ and $v$ be its neighbor. If $f(u) = 1$, then $T_{n+1}^\ast\backslash u$ is a tree of order $n$ and $ f $ is a bad function for this tree. Clearly, $\beta(T_{n+1}^\ast\backslash u)\ge a_{n+1} - 1$. Thus we obtain that $a_{n+1} - a_n \le 1$, and so we are done.\\ 
Now, suppose that $f(u) = -1$. There exists $w\in N(v)$ such that $f(w)=1$, because otherwise we change the value of $u$ to $1$ and obtain a contradiction. Now, change the values of $u$ and $w$ to $1$ and $-1$, respectively. Call this new function by $g$. Obviously, $ g $ is a bad function for $ T_{n+1}^\ast $ and $g(V(T^\ast_{n+1})) = a_{n+1}$. Now, by the previous case the proof is complete.}
\end{proof}

\begin{lem}
\label{lemtree}
{
 For a positive integer $k$, let $T$ be a tree with a bad function $f$ in which exactly $k$ vertices have value $-1$. Then $|V(T)| \le 10k-6$ and the equality holds, if and only if $T$ is constructed as follows:\\ \\
 Let $T^\prime$ be a tree of order $k$. For every vertex $v \in V(T^\prime)$, add a set of new vertices of size $d(v) + 1$ and join $v$ to all of them. Then for each vertex $w$ of this set, add two new vertices and join $w$ to these vertices.
}
\end{lem}

\begin{proof}
{Consider a tree constructed as given in the statement of lemma. Assign the value $-1$ to the vertices of $T'$ and assign the value $1$ to other vertices. Call this function by $f$. Clearly, $f$ is a bad function.\\ 
Now, consider a tree $T^\ast$ of maximum order which satisfies the hypothesis of the lemma. Let $f$ be a bad function for this tree. Call the induced subgraph on the vertices of $T^\ast$ with value $-1$ by $F$. We prove that $F$ is connected. By contradiction, assume that $F$ has at least two connected components. Let $u$ and $v$ be two vertices in different components of $F$. Consider the path between $u$ and $v$ and call it by $P$. There exists a vertex $w\in V(P)$ such that $d(u,w)=min\,d(u,x)$, where $x\in V(P)$ and $f(x)=1$. Let $w'\in V(P)$, such that $d(u,w')=d(u,w)-1$. Obviously, $f(w')=-1$. Now, remove the edge $w'w$ and join $u$ to $v$. If there exists $t \in {N(w) \setminus V(P)}$ and $f(t)=1$, then remove the edge $tw$ and join $w'$ to $t$, otherwise change nothing.  It is not hard to see that $f$ is a bad function for this new graph, which is clearly a tree. Note that $f(N(v))$ has decreased by $1$. Now, add a vertex $z$ to this tree and join it to $v$. Call this tree by $T'$. Consider a function $g$ for $T'$ as follows:
$$
g(x)= \left\{
\begin{array}{lll}
f(x)&&\mbox{if $x\neq z$}\\
1&&\mbox{if $x=z.$}\\
\end{array}
\right.
$$ 
It is clear that $g$ is a bad function for $T'$. This contradicts the maximality of $T^\ast$. So $F$ is connected.

Clearly, each vertex $u\in F$ has at most $d_F(u) + 1$ neighbors with value $1$. Each vertex with value $1$ in $T^\ast$ has at most one neighbor with value $-1$, otherwise we have a cycle. Therefore, each vertex with value $1$ which has a neighbor in $F$, has at most two other neighbors with value $1$. Indeed, these two neighbors are pendant, otherwise $T^*$ has a cycle. Therefore, the following holds:
$$
|V(T^*)| \le k + \sum_{u \in F}{(d_F(u)+1)} + 2 \sum_{u \in F}{(d_F(u)+1)} = k + 3 (2k - 2 + k) = 10k - 6
.$$
It is clear that if $|V(T^*)|=10k-6,$ then $T^*$ is constructed as the pattern discussed in lemma. So the proof is complete. 
}
\end{proof}

In the next theorem, we determine the exact value of $a_n$.
\begin{thm} 
{
For every positive integer $n \ge 3$, $a_n = n - 2\lceil \frac{n+6}{10} \rceil $.
}
\end{thm}

\begin{proof}
{Assume that $n = 10k-6 $, for some positive integer $k $, and $T $ is a tree of order $n$ for which $ \beta(T) = a_n$. So by Lemma \ref{lemtree}, every bad function for $ T $ assigns $-1$ to at least $ k $ vertices. Hence, the assertion is proved for $ n = 10k-6$.
Now, suppose that $ n = 10k-5 $, for a positive integer $ k $ and $ T $ is a tree for which $ \beta(T)=a_n$. By Lemma~\ref{lemtree}, every bad function for $ T $ assigns $-1$ to at least $ k + 1 $ vertices. So, $ \beta(T) \le 10k - 5 - 2(k + 1)$. Let $T'$ be a tree of order $10k-6$, where $\beta(T') = 8k-6$.
Add a vertex with value $ -1 $ and join it to an arbitrary vertex of $ T $. Clearly, the bad decision number of this tree is $8k -7 = n - 2\lceil \frac{n+6}{10} \rceil $.\\
Now, since for each positive integer $k$, $a_{(10(k+1)-6)}-a_{(10k-5)}=9$, by Lemma $\ref{lemstep}$, we obtain that for each $10k-5\leq n< 10(k+1)-6$, $a_{n+1}-a_{n}=1$. So the proof is complete.
}
\end{proof}

\begin{thm}
{
For every connected graph $G$, $ \beta(G) \ge |V(G)| - |E(G)| - 1$.
}
\end{thm}

\begin{proof}
{ The proof is by induction on $ |E(G)| - |V(G)|$. If $G$ is a cycle, then clearly, the inequality holds. So assume that $G$ is not a cycle. Since $ G $ is connected, $ |E(G)| - |V(G)| \ge -1$. If $ |E(G)| - |V(G)| = -1$, then $ G $ is a tree and so by Theorem $ 7 $ in \cite{sub}, we are done. Suppose that the assertion holds for every graph $H$, where $ |E(H)| - |V(H)| \leq k$. Let $ G $ be a connected graph where $ |E(G)| - |V(G)| = k + 1$. Since $ |E(G)|- |V(G)| \ge 0$, $ G $ contains a cycle $ C,$ and there exits a vertex $v$ such that $ v \in V(C) $ and $ d(v) \ge 3$. Assume that $ u,w \in N(v) \cap V(C)$. Let $ x \in N(v)\backslash \{ u,w \}$. Remove the edge $ vw $ and $xv$ and add a new vertex $ v'$. Join $ v' $ to $ x$ and $w$. Call this new graph by $ G'$. Clearly, $G'$ is connected. Since $ |E(G')| = |E(G)| $ and $ |V(G')| = |V(G)| + 1$, we obtain that $|E(G')| -  |V(G')| = k$. By induction hypothesis, $\beta(G') \ge -k-1$. Let $f$ be a bad function such that $f(V(G'))=\beta(G')$. We provide a bad function for $G$. Define the function $ g $ as follows: For every vertex $z\in V(G)\setminus \{v\}$, $g(z)=f(z)$. If $f(v)=f(v')=1$, then define $g(v)=1$. Otherwise, $g(v)=-1$. Clearly, $ g(V(G)) \ge f(V(G')) - 1 \ge -k-2$, as desired.}
\end{proof}
Now, we show that the bad decision number of every cubic graph is non-negative.
\begin{thm}
{
If G is a cubic graph, then $ \beta(G) \ge 0$.
}
\end{thm}
\begin{proof}
{
Let $G$ be a cubic graph of order $n$. By Theorem $ 30 $ in \cite{hen}, we have a total dominating set $ S $ of size at most $\frac{n}{2}$. Assign value $-1$ and $1$ to all vertices of $ S $ and $V\setminus S$, respectively. Thus, $ \beta(G) \ge 0$.
}
\end{proof}
\begin{thm}\label{mmm}
{
Let $ G $ be a cubic graph of order $n$. The following three statements are equivalent:

$i)$ The vertices of $ G $ can be partitioned into the graph shown in Figure~\ref{fig:decompose}.

$ii)$ There exists a bad function $f$, such that $f(N(v))=1$, for every $v\in V(G)$.

$iii)$ $\beta(G) = \frac{n}{3}$.

\begin{figure}[h!]
\centering
  \includegraphics[height=1cm]{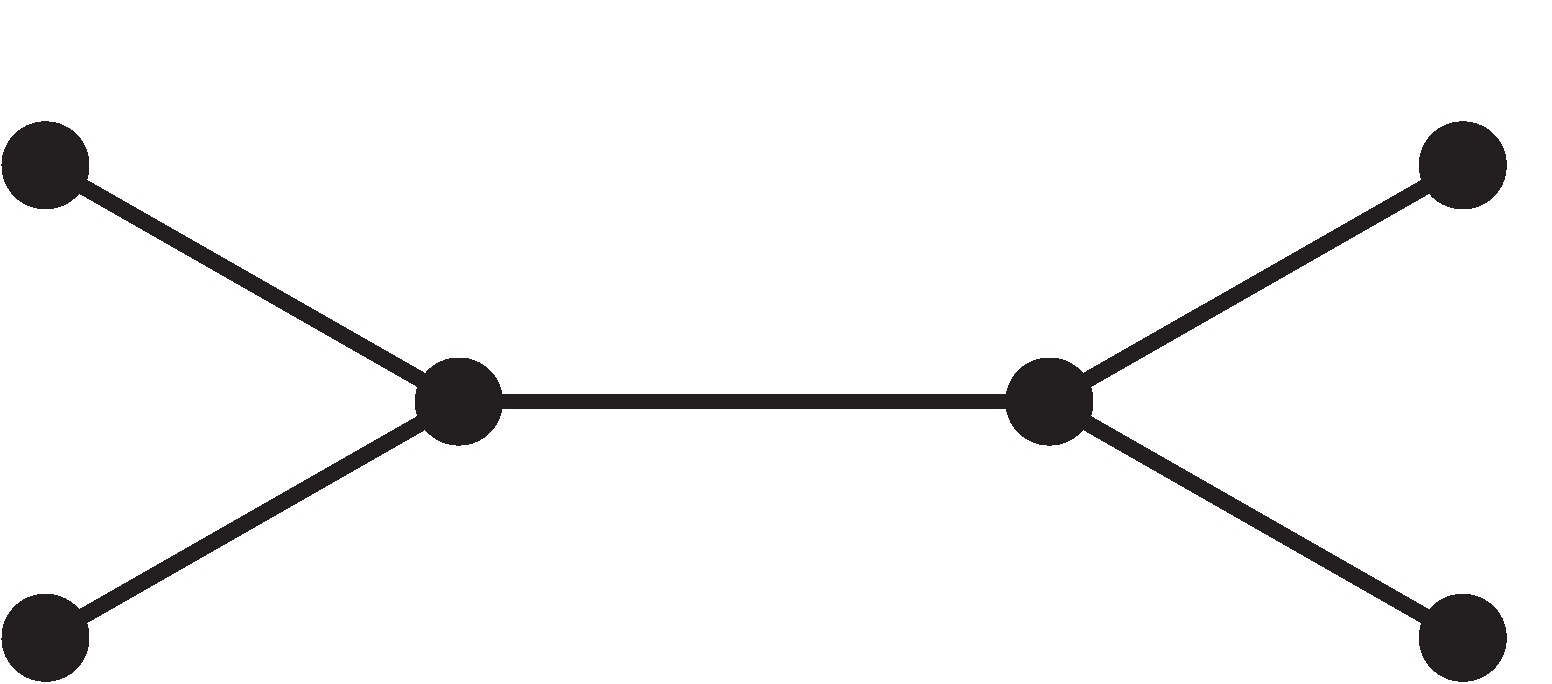}
  \caption{The graph $H$}
  \label{fig:decompose}
\end{figure}
}
\end{thm}

\begin{proof}
{
$(i) \Rightarrow (ii)$ Let $H$ be the graph shown in Figure~\ref{fig:decompose}. Suppose that $ V(G) $ is partitioned into $m$ graphs $H_1,\dots , H_m,$ for some positive integer $m$, where $H_i \simeq H$, for $i = 1,\ldots ,m$. Now, define a bad function $f$ for $G$ as follows. Assign the value $-1$ to each $u\in V(H_i),$ where $d_{H_i}(u)=3$, for $i = 1,\ldots, m.$ Also assign $1$ to other vertices of $G$. Obviously, $f(N(v)) = 1$, for each $v\in V(G)$.

$(ii) \Rightarrow (i)$ Now, assume that there exists a bad function $ f $ such that $ f(N(v))=1 $, for each $ v\in V(G)$. Each vertex with value $ -1 $, is adjacent to exactly one vertex with value $ -1 $. Consider the induced subgraph on $N[u_1]\cup N[u_2]$, for each pair of adjacent vertices $u_1$ and $u_2$ with $f(u_1)=f(u_2)=-1,$ and call them by $H_1,\dots , H_m,$ for some positive integer $m$. Clearly, $H_i \simeq H$, for each $i=1,\ldots, m.$ It is not hard to see $V(H_i)\cap V(H_j) =\emptyset $, for each $i,j$, $1\leq i < j\leq m,$ and $V(G)=\bigcup_{i=1}^{m}V(H_i)$.

$(ii) \Rightarrow (iii)$ It is not hard to see that $\sum\limits_{v\in V(G)} f(N(V))=3f(V(G))=n $. So by Theorem $4$ in $[6]$, we obtain that $\beta(G) = \frac{n}{3}$.

$(iii) \Rightarrow (ii)$ Consider a bad function $f$ for $G$ such that $f(V(G))= \beta(G) = \frac{n}{3}.$ Clearly, $ 3f(V(G))=\sum\limits_{v\in V(G)} f(N(V))= n. $ Thus $f(N(V))=1,$ for each $v\in V(G)$.
}
\end{proof}
\begin{comment}
\bf{Question. }\rm For every $2k$-regular graph $G$ of order $n$, $ \beta(G) \ge -\frac{n}{2k+1}$.
\begin{rem}
{
We present an infinite family of graphs for which the bad decision number is $-\frac{n}{2k+1}$. Consider $ K_{2k+1}$, $k\geq 3,$ and remove the edge $uv$ and call this graph by $F$. We claim that for every bad function $f$ for this graph, there exist at least $ k + 1 $ vertices with value $ -1$. By contradiction, assume that $F$ has at most $ k $ vertices with value $ -1$. Obviously, there exists a vertex $z\in V(F)\setminus \{u,v\}$ for which $ f(z)=-1 $. Note that $d_{F}(z)=2k$, and it has at most $ k-1 $ neighbors with value $ -1$, a contradiction. Thus, $ \beta(F) \leq -1$. Let $s$ be a positive integer. Consider disjoint union of $s$ copies of $F$ and call them by $F_1,\ldots, F_s$, where $u_i, v_i \in V(F_i)$, for $ i = 1,\dots ,s$, are corresponding vertices of $u$ and $v$, respectively. Now, join $v_i$ to $u_{i+1}$ modulo $s$, for $i=1,\ldots ,s$. Call this graph by $G$. Clearly, $|V(G)|=(2k+1)s$ and $ \beta(G) = -\frac{n}{2k+1}$.
\begin{figure}[h!]
\centering
  \includegraphics[height=4cm]{2k}
  \caption{Example for $k=3$}
  \label{fgr:example}
\end{figure}
}
\end{rem}
\end{comment}

%%%%%%%%%%%%%%%%     NICE DECISION NUMBER

\section{Nice Decision Number} 
Here, we provide a lower bound for the nice decision number of trees and show that $\overline{\beta}(G)\geq -\frac{3n}{7}$, for every cubic graph $G$ of order $n$. We start this section with the following result.

In Theorem 10 in \cite{henning}, Henning showed that the nice decision number of every tree is non-negative. We prove this by a simpler approach.
\begin{thm}
{For every tree $ T $, $\overline{ \beta}(T)\geq 0 $.
}
\end{thm}

\begin{proof}
{We apply induction on $ n =|V(T)| $. For $ n=1,2, $ the assertion is trivial. Suppose that the assertion holds for every tree of order at most $n-1$. Consider a tree $ T $ of order $ n \geq 3 $. Two cases maybe considered:\\ \\
{\bf Case 1.} Suppose that there are two pendant vertices $v_1$ and $v_2$, with a common neighbor $u$. Let $ T^{\prime}=T\setminus \{v_{1} , v_{2}\} $. By induction hypothesis, $ \overline{\beta}(T^{\prime})\geq 0 $. Consider a nice function $ f $ such that $ f(V(T^{\prime}))=\overline\beta(T^{\prime}) $. Now, we introduce a nice function $ g $ for $ T $ as follows:

If $  v\in V(T)\setminus \{u,v_{1},v_{2}\} $, then define $g(v)=f(v)$. Assign values $-1, 1$ and $f(u)$ to $u, v_1$ and $v_2$, respectively. It is not hard to see that $ g $ is a nice function and $ g(V(T))\geq 0 $. So we are done.\\ 
\\
{\bf Case 2.} Suppose that there is no pair of pendant vertices with a common neighbor. Let $P$ be one of the longest paths in $T$, and $u$ be a pendant vertex in $V(P)$. Define $ T'=T\setminus  N[u]$. It is straightforward to See that $T'$ is a tree. Consider a nice function $ f $ for $ T' $ such that $ f(V(T'))=\overline{\beta}(T') $. Now, define a function $ g $ for $ T $ as follows:

$$
g(v) = 
\left\{
\begin{array}{lr}
f(v)& v \in V(T)\setminus  N[u] \\
1 & v=u\\
-1 & v\in N(u)\\
\end{array}
\right.
$$

It is clear that $ g $ is a nice function and $ g(V(T))\geq 0 $, so the proof is complete.
}
\end{proof}
\begin{thm}
{Let $G$ be a cubic graph of order $n$. If $\overline{\beta}(G)=0$, then $4\,|\,n$.
}
\end{thm}
\begin{proof}
{We note that $n \equiv 0$ (mod 2). Suppose that $n = 2k$, for some positive integer $k$. Let $P$ and $N$ be the set of vertices with values $1$ and $-1$, respectively. Since $\overline{\beta}(G) = 0$, we conclude that $|P| = |N| = k$. Clearly, $d_P(v) \leq 2,$ for every $v \in N$, and $d_N(u) \geq 2$, for every $u \in P$. Let $e$ be the number of edges between $P$ and $N$. So we obtain that $2|P| \leq e \leq 2|N|$. Hence, $e = 2|P|$. Therefore, $d_P(v) = 2$, for every $ v \in N$. Thus, $d_N(v) = 1$ for every $ v \in N$. So we obtain that $|N| = 2t$, for some positive integer $t$. Thus, $n = 2k = 2|P| = 2|N| = 4t$. The proof is complete.}
\end{proof}

Now, we introduce a lower bound for cubic graphs.

\begin{thm}
For every cubic graph $ G $ of order $ n $, $ \overline{\beta}(G)\geq \frac{-3n}{7} $.
\end{thm}
\begin{proof}
{
Let $ f $ be a nice function of $ G $ for which $ f(V(G))=\overline{\beta_{D}}(G) $, and $ f $ has the minimum number of edges with both end vertices having value $ -1 $. Consider the following sets:
$$N=\{v\in V(G)\,|\,f(v)=-1\}$$
$$P=\{v\in V(G)\,|\, f(v)=1\}$$
$$N_{i}=\{v\in N\,|\,f(N[v])=i\} \:\text{ for }\, i=0, -2, -4 $$
$$P_{i}=\{v\in P\,|\,f(N[v])=i\} \:\text{ for }\, i=0, -2.$$
Now, we prove the following statement:\\
$(i)$ Each vertex in $ N_{-2}\cup N_{-4} $ has at least one neighbor in $ N_{0} $.

It is not hard to see that each vertex in $ N_{-2} \cup N_{-4} $, say $ v $, is adjacent to some vertex in $ P_{0}\cup N_{0} $, otherwise we can change the value of $ v $ from $ -1 $ to $ 1 $ to increase $ f(V(G)) $, a contradiction.
Clearly, $(i)$ is obvious for every vertex in $ N_{-4} $.\\
Assume that $ v \in N_{-2} $ has no neighbor in $ N_{0} $. Since $ f(N[v])=-2 $, there is exactly one vertex in $ P_{0} $, say $ u $, adjacent to $ v $. Call the other neighbors of $ v $, $ u_{1} $ and $ u_{2} $. Clearly, $u_1,u_2 \in  N_{-2}\cup N_{-4} $. Now, define the function $g$ as follows:
$$
g(x)= \left\{
\begin{array}{ccc}
f(x)&\text{if }x\neq u,v\\
-f(x)&\text{if }x=u,v\\
\end{array}
\right.
.$$
\begin{figure}[h!]
\centering
  \includegraphics[height=3cm]{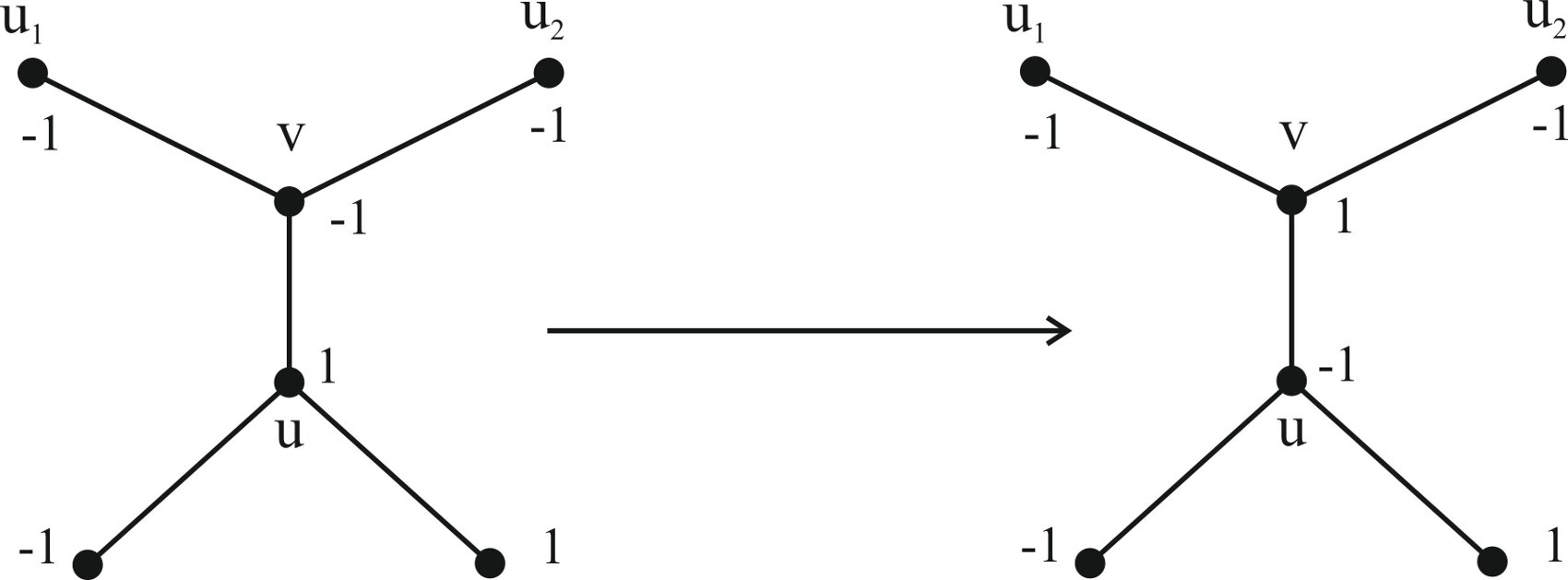}
  \caption{The function $g$ on $G$}
  \label{fig:3n71}
\end{figure}

For each $ x\notin \{u_{1},u_{2},v,u \}$, $ g(N[x]) \leq f(N[x]) $. Also $ g(N[x])=f(N[x]) $ for $ x=u,v $, and $ g(N[x])=f(N[x])+2 $ for $ x=u_{1},u_{2} $. Since $ u_{1},u_{2} \notin N_{0} $, $ f(N[x])\leq -2 $ for $ x=u_{1},u_{2} $. Thus $g(N[x])\leq 0$, for each $ x\in V(G) $. For the function $g$, the number of edges with endpoints $-1$ is reduced one unit, a contradiction. So $(i)$ is proved.

Now, consider the following sets:
$$M_1=\{z\in N_{-4}\,|\,|N(z)\cap N_0|=1\}\,\,\,\,\,\,\, \,\,\,M_2=\{z\in N_{-4}\,|\,|N(z)\cap N_0|\geq 2\}$$
By $(i)$, we obtain that $M_1 \cup M_2 = N_{-4}$. If $t$ denotes the number of edges between $N_0$ and $N_{-2}\cup N_{-4}$, then $t \leq |N_0|$ and $t \geq 2|M_2|+|M_1|+|N_{-2}|$. Thus, $2|M_2|+|M_1|+|N_{-2}|\leq |N_0|$ (1).
Let $U=\{z\in N_0\,|\,\,|N(z)\cap M_1|=1\}$, and $U_i=\{z\in U\,|\,\,|N(z)\cap P_0|=i\}$, for $i=0,1,2$.

We prove the following two statements:\\
$(ii)$ There is no pair of vertices in $U_0$ with a common neighbor in $P$.

By contradiction, suppose that there exist $u_1, u_2\in U_0$ and $w\in N(u_1)\cap N(u_2)\cap P$. By definition of $U_0$, $w\in P_{-2}$ and there exist $v_1, v_2\in M_1$ and $w_1\in P_{-2}$ such that $w_1,v_1\in N(u_1)$ and $v_2\in N(u_2)$. Since $v_1, v_2\in M_1$,   we obtain that $v_1\neq v_2$. Now, define the function $g$ as follows:

$$
g(x)= \left\{
\begin{array}{ccc}
f(x)&\text{if }x\neq u_1,w,v_2\\
-f(x)&\text{if }x=u_1,w,v_2.\\
\end{array}
\right.
$$
\begin{figure}[h!]
\centering
  \includegraphics[height=2cm]{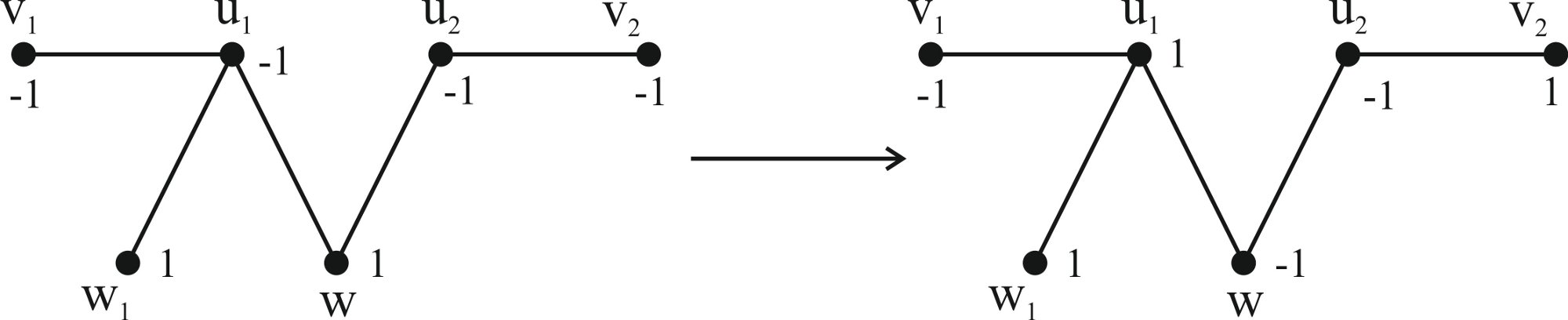}
  \caption{The function $g$ on $G$}
  \label{fig:3n72}
\end{figure}

It is straightforward to see that $g(N[x]) \leq f(N[x])$ for each $x\notin N(u_1)\cup N[v_2] $. Also $g(N[v_2]) = -2$ and $g(N[w])=f(N[w])$. Since $w_1\notin N(v_2)$ and $w_1\in P_{-2}$, we have $g(N[w_1])=0$. Clearly, $g(N[v_1])\leq f(N[v_1])+4\leq 0$. It is not hard to see that $g(N[x])=f(N[x])+2\leq 0$, for each $x\in N(v_2)\setminus \{v_1,u_2\}$. Thus $g$ is a nice function and $g(V(G)) = f(V(G)) + 2$, a contradiction.\\ \\
$(iii)$ There is no pair of vertices $u_1\in U_1$ and $u_2\in U$ with a common neighbor in $P_0$.

By contradiction, assume that there exist $u_1\in U_1$, $u_2\in U$, and $w\in N(u_1)\cap N(u_2)\cap P_0$. By definition, there exist $v_1, v_2\in M_1,\, v_1\neq v_2,$ which are adjacent to $u_1$ and $u_2$, respectively. Let $g$ be the function defined in $(ii)$. Note that $g(N[w])=f(N[w])$. Similar the argument given in $(ii)$, $g(N[x])\leq 1$, for every $x\in V(G)\setminus \{w\}$. Thus $g$ is a nice function and $g(V(G)) = f(V(G)) + 2$, a contradiction.

Now, we prove the theorem.
Clearly, by $(ii)$, $|U_0|\leq |P_{-2}|/2$. If $V_1$ denotes the set of vertices in $P_0$ which are adjacent to a vertex in $U_1$, then by $(iii)$, $|V_1|\geq |U_1|.$ Also if $V_2$ denotes the set of vertices in $P_0$ which are adjacent to a vertex in $U_2$, then $|V_2| \geq |U_2|$. By $(iii)$, $V_1\cap V_2=\emptyset$, so we obtain that $|U_1|+|U_2|\leq |P_0|$. It is clear that $|M_1|=|U|$. Hence, $|M_1|\leq |P_{-2}|/2+|P_0|\text{ }(2)$.

If $k$ denotes the number of edges between $P$ and $N$, then $|N_{-2}|+2|N_0|=k=3|P_{-2}|+2|P_0|\text{ }(3)$.\\
By adding $(1), (2)$ and $(3)$, we obtain that $2|N| \leq 3|P| +|P_{-2}|/2 +|N_0|$. By $(3)$, $|N_0| \leq \frac{3}{2}|P_{-2}|+|P_0|$. Thus $2|N|\leq 5|P|-|P_0|\leq 5|P|$. So we obtain that $\overline{\beta}(G)=2|P|-n\geq \frac{-3n}{7}$.
}
\end{proof}

{It seems that for every cubic graph $ G $ of order $ n $, $ {\overline{\beta}}(G) \geq\frac{-n}{5} $.
We introduce an infinite family of bipartite graphs for which the equality holds.\\
For a natural number $ n $, consider $ 2n $ disjoint copies of $K_{2,3}$, and call them by $H_1, H_2, \ldots ,H_{2n}$. Let $H_i = (\{a_i, b_i\},\{ c_i,d_i, e_i\})$. Now, join $d_{2i-1}$ to $d_{2i}$, for $1 \le i \le n$, and join $c_i$ to $e_{i+1}$, modulo $2n$ for $1 \le i \le 2n$. Call this graph by $G$. Consider a nice function $ f $ for $ G $, where $ f(V(G))=\beta(G) $. We prove that for every $ i=1,\ldots,2n $, $ f(V(H_{i}))\leq -1 $. If $ f(a_i)=-1 $, then since $ f(N[b_i])\leq 0 $, we are done. So assume that $ f(a_i)=f(b_i)=1 $. Clearly, $ f(c_i)=f(d_i)=f(e_i)=-1 $. So the proof is complete.
\begin{comment}
\begin{figure}[h!]
\centering
  \includegraphics[height=3cm]{1}
  \caption{The graph $H$}
  \label{fig:nfrac5}
\end{figure}
\end{comment}
}

%%%%%%%%%%%%%%%%%%%%%%%%%%%%%%%% GOOD DECISION NUMBER

\section{Good Decision Number}
In this section, we show that $\gamma(G)\leq \frac{5n}{7}$, for every cubic graph of order $n$.

Now, we present some results on good decision number of cubic graphs.
\begin{thm}
\label{bipartite}
{For every cubic bipartite graph $G$ of order $n,\,\gamma(G)\leq\frac{5n}{7}$.
}
\end{thm}
\begin{proof}
{Let $G$ have bipartite partition $(U,V)$. We assign values $ -1$ and $ 1 $ to the vertices of $G$ by the following algorithm:\\ \\
Step 1. Consider two adjacent vertices $u$ and $v$ which have no value. Assign $-1$ to both $u$ and $v$. If $w\in V(G)$ has no value and $\min{(d(v,w), d(u,w))}\leq 2 $, then assign $1$ to $w$. Note that $G$ has at most $12$ vertices with this property. Do this procedure as much as possible. \\ \\
Step 2. Let $v\in V(G)$ be an arbitrary vertex which has no value. Assign $-1$ to $v$ and $1$ to each $w\in V(G)$ with no value, where $d(v,w)=2$. Note that there exist at most three vertices with this property. Do this procedure as much as there exists a vertex with no value.

Call this function by $f$. We show that each pair of vertices with value $-1$ have no common neighbor. Suppose that there exists $v\in V(G),$ such that $v$ has two neighbors $u $ and $w$ with value $-1$. In the algorithm, one of $u$ and $w$ received value prior to the other one, say $u$. Since $d(u,w)=2,\, w$ has value $1$, a contradiction. So $f$ is a good function.
It is straightforward to see that $f(V(G))\leq \frac{5n}{7}$. So $\gamma(G)\leq \frac{5n}{7}$ and the proof is complete.
}
\end{proof}

\begin{thm} 
\label{lemma14}
{If for every cubic bipartite graph $ H $, $ \gamma(H)\leq c|V(H)| $, and $ c>0 $, then for every cubic graph $ G $, $ \gamma(G)\leq c|V(G)|. $
}
\end{thm}
\begin {proof}
{Let $V(G) = \{v_1,\ldots,v_n\}$. We construct a bipartite graph $H$. Consider two copies of vertices of $G$, say $G_{1}$ and $G_{2}$ with vertex sets $ \{v'_{1}, \ldots ,v'_{n}\} $ and $ \{v''_{1}, \ldots ,v''_{n}\} $, respectively. Join $v'_{i}\in V(G_{1})$ to $v''_{j}\in V(G_{2})$, if $v_{i}$ and $v_{j}$ are adjacent. Now, consider a good function $f$ for $H$, where $f(V(H))=\gamma(H).$ With no loss of generality, assume that $f(V(G_{1}))\leq\frac{1}{2}f(V(H)).$ Define a function $g$ as follows: $g(v_i) = f(v'_i),$ for every $i= 1,\ldots ,n$. Clearly, $g$ is a good  function. So, $\gamma(G)\leq\frac{\gamma(H)}{2}$. Since $|V(H)|=|V(G)|/2$, we obtain that $\gamma(G)\leq c|V(G)|.$
}
\end {proof}

Now, by Theorem \ref{bipartite} we have the following corollary:
\begin{cor}
\label{cor 5div7n}
For every cubic graph $G$ of order $ n,\,\gamma(G)\leq\frac{5n}{7}.$
\end{cor}

If $G$ is a cubic graph of order $n$ then $\gamma(G) \geq \frac{n}{3}$ (A more generalized version is discussed in \cite{zelinka}). 

\begin{rem}
We present an infinite family of bipartite cubic graphs, $\{G_{i}\}_{i\ge 1}, $ for which $\gamma (G_i) = \frac{5}{7}|V(G_i)|$, where $|V(G_i)| = 14 \times 2^{i-1}.$ Let $G_1$ be the Heawood graph and call the function shown in Figure~\ref{fig:heawood} by $f_1$. 
\begin{figure}[h!]
\centering
  \includegraphics[height=4.5cm]{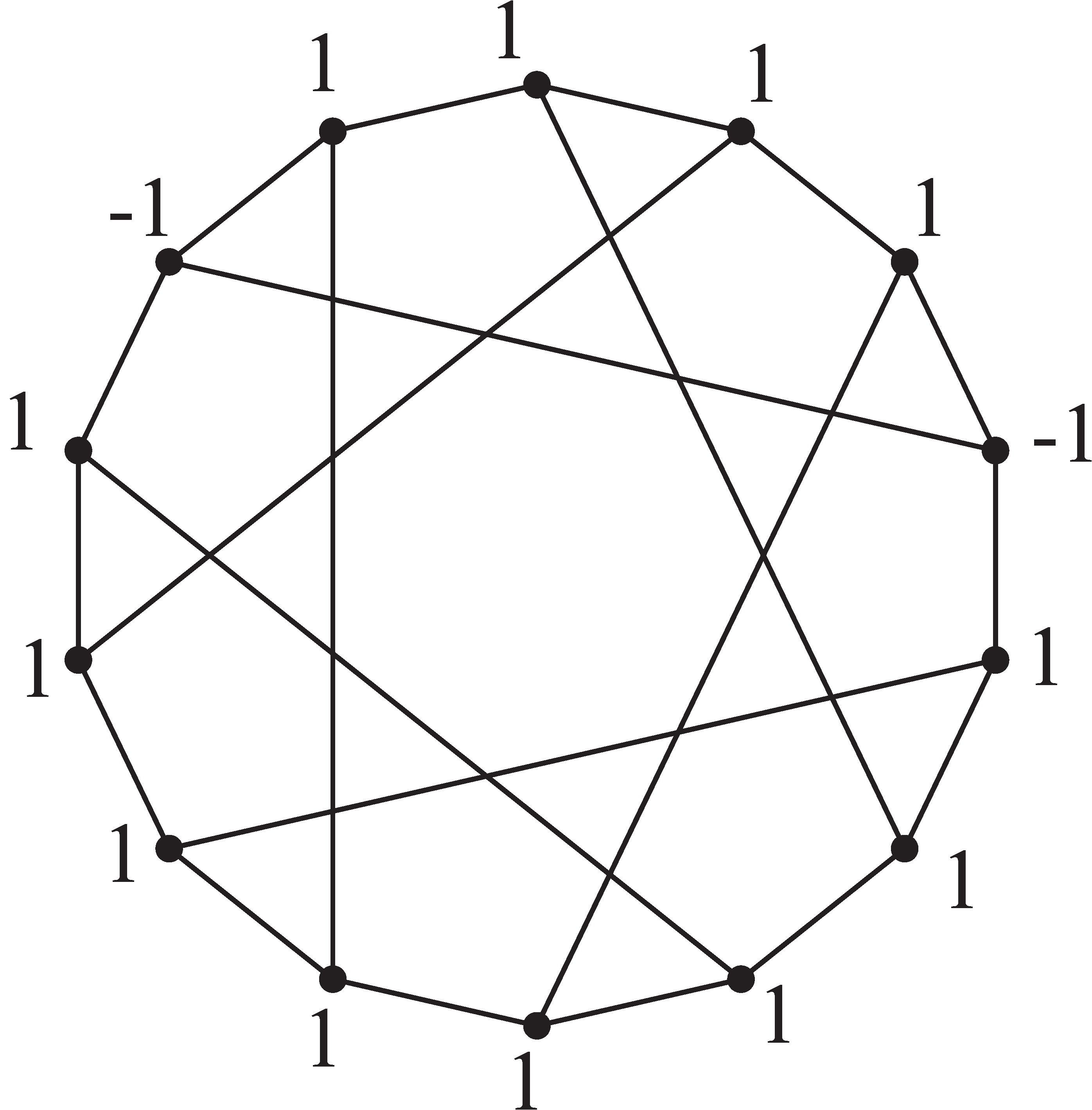}
  \caption{The Heawood graph}
  \label{fig:heawood}
\end{figure}
It is not hard to see that $f_1$ is a good function, and $f_1(V(G_1)) = 10$. For every $ i\geq 2$, construct the bipartite cubic graph $G_i$ as follows:\\
Let $V(G_{i-1})=\{ v_1,\ldots, v_n\}$. Consider two copies of the vertex set of $G_{i-1}$, say $X'_i=$$\{v'_1,\ldots,v'_n\}$ and $X''_i=\{v''_1,\ldots,v''_n\}$. Join $v'_j $ to $v''_k $, if $v_j v_k\in E(G_{i-1})$. Assume that $f_{i-1}$ is a nice function for $G_{i-1}$ such that $f_{i-1}(V(G_{i-1})) = \frac{5}{7}|V(G_{i-1})|$. Define $f_i(v'_i)=f_i(v''_i)=f_{i-1}(v_i)$, for $i=1,\ldots,n$. It is not hard to see that $f_i$ is a good function for $G_i$, and $f_i (V(G_i)) = \frac{5}{7}|V(G_i)|$. Now, assume that $g$ is a good function for $G_i$, where $g(V(G_i))=\gamma (G_i) < \frac{5}{7}|V(G_i)|$. By the Pigeonhole Principle, one of $g(X'_i)$ and $g(X''_i)$ is less than $\frac{5}{14}|V(G_i)|$. Without loss of generality, assume that $g(X'_i)<\frac{5}{14}|V(G_i)|$. Define a function $h$ as follows: $h(v_i)=g(v'_i)$, for every $i=1,\ldots n$. Clearly, $h$ is a good function for $G_{i-1}$. So $\gamma (G_{i-1}) < \frac{5}{7}|V(G_{i-1})|$, a contradiction. Thus $\gamma (G_i) = \frac{5}{7}|V(G_i)|$, for each positive integer $i \geq 1$.
%Therefore, by Corollary~\ref{cor 5div7n}, $\gamma (G_i) = \frac{5}{7}|V(G_i)|$, for each positive integer $i \geq 1$.

\end{rem}

%%%%%%%%%%%%%%%%%%%%%%%%%%%%%%%%%%%%%%%%%%%%%%%%%%%%%%%
%%%%%%%%%%%%%excellet decision
\section{Excellent Decision Number}

In this section, we show that for every tree $T$ of order $n$, $\overline{\gamma}(T) \ge n - 2\lceil\frac{n-4}{3}\rceil $. We prove that for every cubic graph $G$ of order $n$, except Petersen graph, $\overline{\gamma}(G)\leq \frac{3n}{4}$.

The following theorem has been stated in \cite{henning}.

\begin{thm}
\label{thm15}
For every positive integer $n>1$, $\overline{\gamma}(P_n) = n - 2\lceil\frac{n-4}{3}\rceil $.
\end{thm}
\begin{comment}
\begin{proof}
{We apply induction on $n$. For $n=2$ and $3,$ the assertion is trivial. Assume that $n \geq 4$. Let $f$ be an excellent function for $P_n:v_1 v_2 \ldots v_n$. Clearly, $f(x) = 1$, for $x\in \{v_1,v_2,v_{n-1},v_{n}\}$. Also, for every three consecutive vertices, there is at most one vertex with value $-1$. Thus, $f(P_{n}) \geq n - 2\lceil\frac{n-4}{3}\rceil $. Now, define the following excellent function $g$ for $P_n$:

$$
g(v_i) = 
\left\{
\begin{array}{lr}
-1& \mbox{if } i\equiv 0 \, ( \bmod{\, 3})$ for $ i=3,\ldots,n-2  \\
1 & \mbox{otherwise.}
\end{array}
\right.
$$

Therefore, $\overline{\gamma}(P_n) = n - 2\lceil\frac{n-4}{3}\rceil $.
}
\end{proof}
\begin{thm}
For every positive integer $n \geq 3$, $\overline{\gamma}(C_n) = n - 2\lfloor\frac{n}{3}\rfloor $.
\end{thm}
\begin{proof}
{
Let $f$ be an excellent function for $C_n:v_1 v_2 \ldots v_n v_1$. Assume that $k$ vertices have value $-1$, for some positive integer $k$. It is clear that $d(u,v)\geq 3$, for every $u,v\in C_{n},$ where $f(u)=f(v)=-1$. Therefore, $ k \leq \lfloor \frac{n}{3} \rfloor$ and $ \overline{\gamma}(C_n) \geq n - 2\lfloor\frac{n}{3}\rfloor $. Now, define the following excellent function $g$ for $C_n$:\\
$$
g(v_i) = 
\left\{
\begin{array}{lr}
-1& \mbox{if } i\equiv 0 \, (\bmod{\, 3}),\, 1 \leq i \leq n \\
1 & \mbox{otherwise.}
\end{array}
\right.
$$
Thus we obtain that $\overline{\gamma}(C_n) = n - 2\lfloor\frac{n}{3}\rfloor $.
}
\end{proof}
\end{comment}
 \begin{thm}
For every tree $T$ of order $n > 1$, $\overline{\gamma}(T) \ge n - 2\lceil\frac{n-4}{3}\rceil $.
\end{thm}
\begin{proof}
{Consider an excellent function $f$ for $T$, where $f(V(T))=\overline{\gamma}(T)$. Let $P$ be the longest path in $T$. Clearly, both end vertices of $P$, say $w_1$ and $w_2$, are pendants. Note that $f(x)=1$, for each $x\in N([w_1])\cup N([w_2]).$ \\
Suppose that there exists a vertex $ v \notin V(P)$ with value $-1 $. Let $u\in N(v)$ be the vertex such that if we remove $v$ from $T$, then $u$ is in the connected component containing $P$. Remove the edge $uv$ and join $v$ to $w_1$. Clearly, $f$ induces an excellent function for this new tree. Note that the length of the maximum path is increased. Now, repeat the previous procedure as much as possible.\\
Now, if there exist two adjacent vertices $u, v \in V(P)$ with value $-1$, then $u$ has at least three neighbors with value $1$. So there exist two vertices, $w$ and $w'$ adjacent to $u$ with value $1$, where $w,w' \notin V(P)$. Remove the edge $uv$ and $uw'$ and join both $w$ and $v$ to $w'$. It is not hard to see that $f$ is an excellent function for this new tree. Note that the length of the maximum path is increased. Continue this procedure until all edges whose both endpoints have value $-1$ are removed.\\
Now, suppose that there exists a vertex $u$ in $V(P)$ which has a neighbor, $w \notin V(P)$. Clearly, $f(w) = 1$. Let $v, z\in N(u)\cap V(P)$. At most one of the vertices $u$ and $v$ has value $-1$. Thus, $f(u) + f(v) \ge 0$. Remove the edge $uv$ and join $v$ to $w$. It is not hard to see that $f(N[w])\geq 1$. Clearly, $f(N[v]) \geq 1$. Since $ f(N[u])\geq f(u)+ f(w)+f(z) $ and $ f(u)+f(z)\geq 0 $, so $f(N[u])\geq 1  $. By this algorithm, we obtain a path of order $n$ with the excellent function $f$. By Theorem~\ref{thm15}, the proof is complete.
}
\end{proof}

A $2$-distance coloring of a graph is a coloring of the vertices such that two vertices at distance at most 2 receive distinct colors.
Now, we would like to present an upper bound for the excellent decision number of cubic graphs.

In Theorem $2$ from \cite{favaron}, Favaron proved the following theorem. Here we present a short proof for this result.
\begin{thm}
{
For every cubic graph $G$ of order $n$, $\overline{\gamma}(G)\leq \frac{3n}{4}$, except Petersen graph.
}
\end{thm}
\begin{proof}
{
By Main Theorem in \cite{maincoloring} we know that, if $G$ is a connected graph with maximum degree $3$ and $G$ is not the Petersen
graph, then there is a $2$-distance coloring of $G$ with $8$ colors.
%Let $g$ be the coloring of Main Theorem in \cite{coloring} on $V(G)$ using %$8$ colors. 
Let $A$ be the largest color class. Assign $-1$ and $1$ to the vertices of $A$ and $V(G) \setminus A$, respectively. Call this function by $f$. Obviously, $f$ is an excellent function, and $f(V(G))\leq \frac{3n}{4}$. Thus, $\overline{\gamma}(G)\leq \frac{3n}{4}$ and the proof is complete.
 }
\end{proof}
As a result of Theorem $12$ in \cite{kang}, $\overline{\gamma}(G)\leq$ $\frac{5n}{7}$, for every cubic graph $G$ of order $n$. In the following remark, we show that this bound is sharp. 

\begin{rem}
\label{remark2}
We introduce an infinite family of planar cubic graphs of order $n$, whose excellent decision number is $\frac{5n}{7}$. Call the graph shown in Figure~\ref{example gamma bar} by $H$, and let $H'$ be the graph shown in Figure~\ref{examplehprime}. It is not hard to see that $\overline{\gamma}(H) = 5$. Let $F$ be a graph with the vertex set $V(H)\cup V(H')$ and the edge set $E(H)\cup E(H')\cup \{vu'\}$. It is straightforward to see that $\overline{\gamma}(F) = 10$. Consider the disjoint union of $s$ copies of this graph and call them by $F_1, \ldots, F_s$. Let $x'_i, y_i,u_i,w_i \in V(F_i)$, $1\le i\le s$, be the corresponding vertices of $x'\in V(H')$ and $y,u,w\in V(H)$. Now, join $u_j$ to $w_{j+1}$, modulo $s$, $1\le j\le s$. Call this graph by $G$.
% Let $f\rightarrow \{1,-1\}$ be a function for $G$ such that the restriction of $f$ to $F_i$ is the excellent function shown in Figure~\ref{labelexample}. 
Define the function $f$ as follows: 
$$
f(z)= \left\{
\begin{array}{lll}
-1&& z=x'_i, y_i\mbox{, for } i=1,\ldots, s\\
1&&\mbox{otherwise.}\\
\end{array}
\right.
$$
Clearly, $f$ is an excellent function and $f(V(G))= \frac{5}{7}|V(G)|.$ Consider an excellent function $g$ for $G$. Note that $g(N_{G}(z))\geq 2$, for every $z\in V(G)$. If $d_{F_i}(z) = 2$, for some $i$, $1\leq i \leq s$, then $g(N_{F_i}(z))\geq 1.$ Thus the restriction of $g$ to $F_i$ is an excellent function, for $1\leq i \leq s.$ So we obtain that $\overline{\gamma}(G)=\frac{5}{7}|V(G)|$.
\begin{figure}[h!]
\centering
  \includegraphics[height=2.5cm]{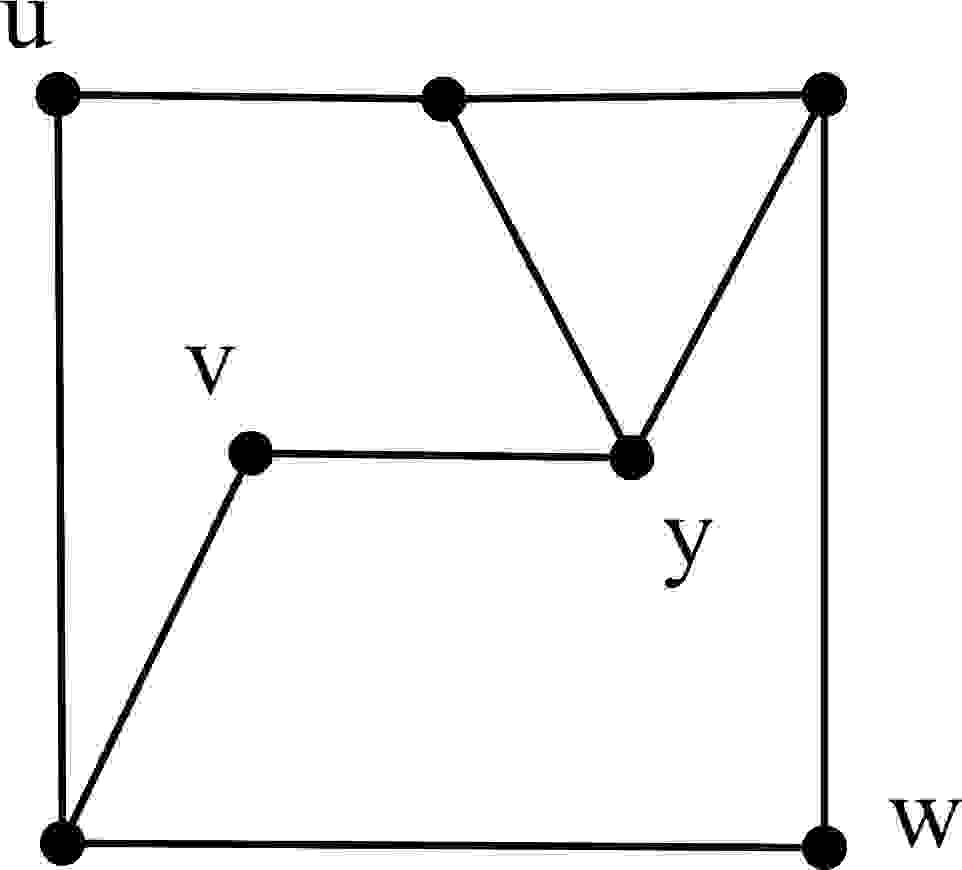}
  \caption{Graph $H$ }
  \label{example gamma bar}
\end{figure}
\begin{figure}[h!]
\centering
  \includegraphics[height=2.5cm]{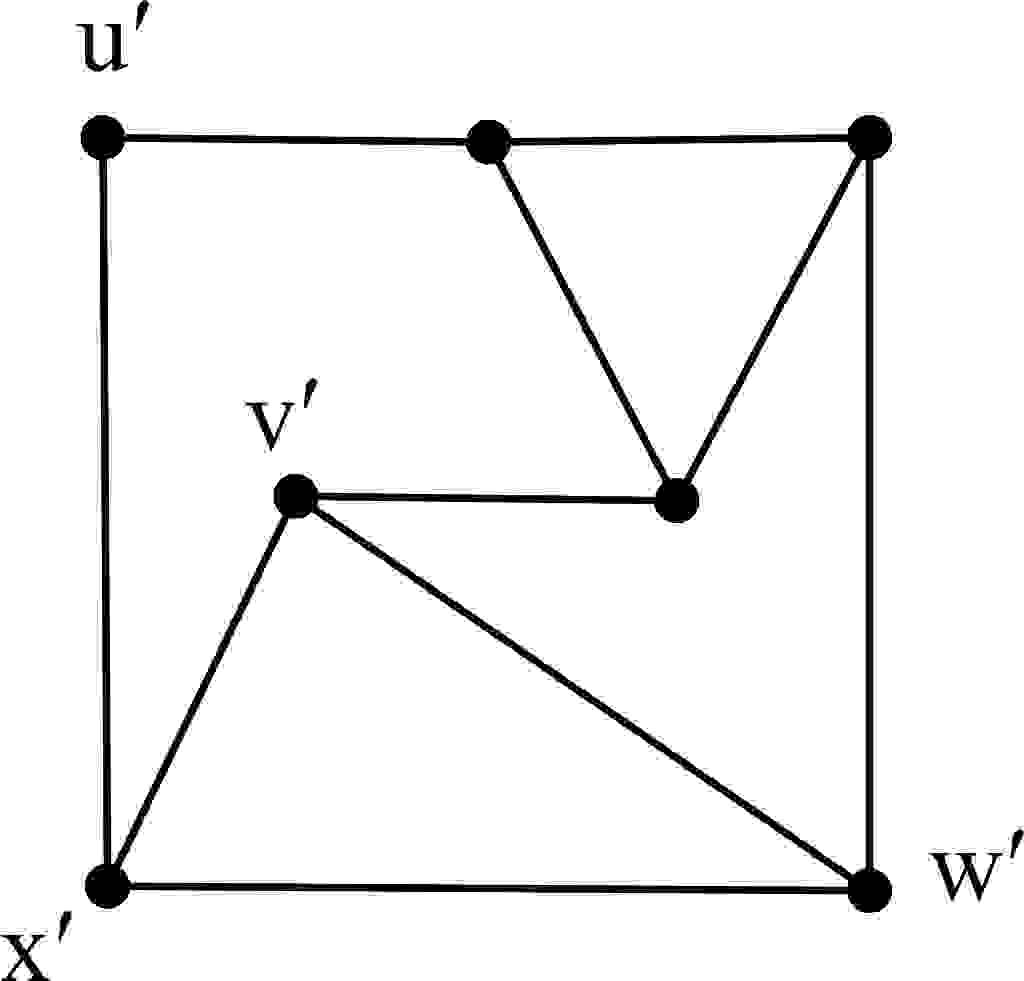}
  \caption{Graph $H'$ }\textbf{}
  \label{examplehprime}
\end{figure}
\begin{figure}[h!]
\centering
  \includegraphics[height=2.7cm]{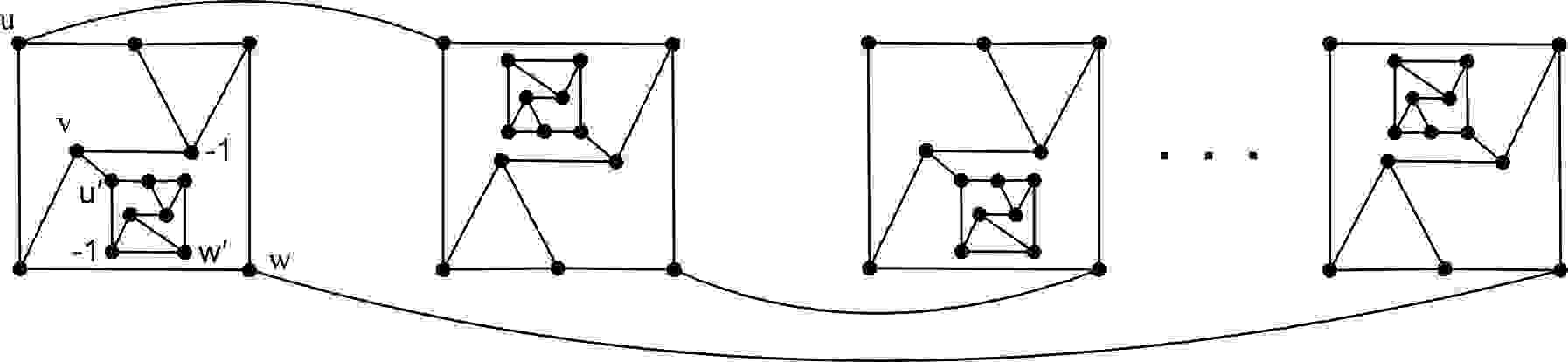}
  \caption{Graph $G$ constructed from $F$}
  \label{fig:example L}
\end{figure}
\end{rem}

\section{Computational Results}
Using a computer search, we obtain some results for decision number of small graphs, see \href{https://github.com/ehsanik/decision_number}{here}. Define $m(G)= |\{H\,|\, \theta(H)=\min\,\theta(G)\mbox{ over all } G\}|$ and $M(G)= |\{H\,|\, \theta(H)=\max\,\theta(G)\mbox{ over all } G\}|$, where $\theta\in\{\beta, \overline{\beta}, \gamma, \overline{\gamma}\}$. We denote $n$, $t(n)$ and $c(n)$ the order of graphs, the number of trees of order $n$ and the number of cubic graphs of order $n$, respectively. \\ \\
% Table generated by Excel2LaTeX from sheet 'Sheet2'
\begin{table}[htbp!]
  \centering

    \begin{tabular}{|c|c|c|c|c|c|}
    \hline
    $n$ & $t(n)$ & $\min \, \beta(G)$ & $m(G)$  & $\max \, \beta(G)$ & $M(G)$ \\
    \hline
    4     & 2     & 0     & 1     & 2     & 1 \\
    5     & 3     & 1     & 3     & 1     & 3 \\
    6     & 6     & 0     & 1     & 2     & 5 \\
    7     & 11    & 1     & 6     & 3     & 5 \\
    8     & 23    & 0     & 3     & 4     & 7 \\
    9     & 47    & 1     & 14    & 5     & 6 \\
    10    & 106   & 0     & 4     & 6     & 7 \\
    11    & 235   & 1     & 36    & 7     & 4 \\
    12    & 551   & 0     & 11    & 8     & 3 \\
    13    & 1301  & 1     & 97    & 9     & 1 \\
    14    & 3159  & 0     & 21    & 10    & 1 \\
    15    & 7741  & 1     & 276   & 9     & 96 \\
    16    & 19320 & 0     & 57    & 10    & 86 \\
    17    & 48629 & 1     & 810   & 11    & 70 \\
    \hline
    \end{tabular}%
    \caption{Bad decision number for trees}
  \label{tab:addlabel}%
\end{table}%

$$
$$

% Table generated by Excel2LaTeX from sheet 'Sheet1'
\begin{table}[htbp!]
  \centering
  
    \begin{tabular}{|c|c|c|c|c|c|}
    \hline
    $n$ & $t(n)$ & $\min \, \overline{\beta}(G)$ & $m(G)$  & $\max \, \overline{\beta}(G)$ & $M(G)$ \\
    \hline
    4     & 2     & 0     & 2     & 0     & 2 \\
    5     & 3     & 1     & 3     & 1     & 3 \\
    6     & 6     & 0     & 3     & 2     & 3 \\
    7     & 11    & 1     & 10    & 3     & 1 \\
    8     & 23    & 0     & 8     & 4     & 1 \\
    9     & 47    & 1     & 33    & 3     & 14 \\
    10    & 106   & 0     & 19    & 4     & 9 \\
    11    & 235   & 1     & 122   & 5     & 5 \\
    12    & 551   & 0     & 58    & 6     & 2 \\
    13    & 1301  & 1     & 471   & 7     & 1 \\
    14    & 3159  & 0     & 177   & 6     & 54 \\
    15    & 7741  & 1     & 1888  & 7     & 27 \\
    16    & 19320 & 0     & 612   & 8     & 13 \\
    17    & 48629 & 1     & 7771  & 9     & 4 \\
    \hline
    \end{tabular}%
    \caption{Nice decision number for trees}
  \label{tab:addlabel}%
\end{table}%
$$
$$
\begin{table}[htbp!]
  \centering

    \begin{tabular}{|c|c|c|c|c|c|}
    \hline
    $n$ & $t(n)$ & $\min \, \gamma(G)$ & $m(G)$  & $\max \, \gamma(G)$ & $M(G)$ \\
    \hline
    4     & 2     & 2     & 1     & 4     & 1 \\
    5     & 3     & 3     & 2     & 5     & 1 \\
    6     & 6     & 2     & 2     & 6     & 1 \\
    7     & 11    & 3     & 5     & 7     & 2 \\
    8     & 23    & 2     & 3     & 8     & 2 \\
    9     & 47    & 3     & 11    & 9     & 4 \\
    10    & 106   & 2     & 6     & 10    & 6 \\
    11    & 235   & 3     & 28    & 11    & 9 \\
    12    & 551   & 2     & 11    & 12    & 15 \\
    13    & 1301  & 3     & 67    & 13    & 25 \\
    14    & 3159  & 2     & 23    & 14    & 42 \\
    15    & 7741  & 3     & 171   & 15    & 70 \\
    16    & 19320 & 2     & 47    & 16    & 123 \\
    17    & 48629 & 3     & 433   & 17    & 213 \\
    \hline
    \end{tabular}%
    \caption{Good decision number for trees}
  \label{tab:addlabel}%
\end{table}%
$$
$$
\begin{table}[htbp!]
  \centering

    \begin{tabular}{|c|c|c|c|c|c|}
    \hline
    $n$ & $t(n)$ & $\min \, \overline{\gamma}(G)$ & $m(G)$  & $\max \, \overline{\gamma}(G)$ & $M(G)$ \\
    \hline
    4     & 2     & 4     & 2     & 4     & 2 \\
    5     & 3     & 3     & 1     & 5     & 2 \\
    6     & 6     & 4     & 2     & 6     & 4 \\
    7     & 11    & 5     & 6     & 7     & 5 \\
    8     & 23    & 4     & 1     & 8     & 10 \\
    9     & 47    & 5     & 4     & 9     & 14 \\
    10    & 106   & 6     & 16    & 10    & 27 \\
    11    & 235   & 5     & 1     & 11    & 43 \\
    12    & 551   & 6     & 7     & 12    & 82 \\
    13    & 1301  & 7     & 42    & 13    & 140 \\
    14    & 3159  & 6     & 1     & 14    & 269 \\
    15    & 7741  & 7     & 12    & 15    & 486 \\
    16    & 19320 & 8     & 99    & 16    & 939 \\
    17    & 48629 & 7     & 1     & 17    & 1765 \\
    \hline
    \end{tabular}%
    \caption{Excellent decision number for trees}
  \label{tab:addlabel}%
\end{table}%
$$
$$
% Table generated by Excel2LaTeX from sheet 'Sheet1'
\begin{table}[htbp!]
  \centering
  
    \begin{tabular}{|c|c|c|c|c|c|}
    \hline
    $n$ & $c(n)$ & $\min \, \beta(G)$ & $m(G)$  & $\max \, \beta(G)$ & $M(G)$ \\
    \hline
    4     & 1     & 0     & 1     & 0     & 1 \\
    6     & 2     & 2     & 2     & 2     & 2 \\
    8     & 5     & 0     & 2     & 2     & 3 \\
    10    & 14    & 2     & 14    & 2     & 14 \\
    12    & 57    & 0     & 1     & 4     & 31 \\
    14    & 341   & 2     & 120   & 4     & 221 \\
    16    & 2828  & 0     & 2     & 4     & 2805 \\
    18    & 30468 & 2     & 82    & 6     & 8166 \\
    \hline
    \end{tabular}%
    \caption{Bad decision number for cubic graphs}
  \label{tab:addlabel}%
\end{table}%
$$
$$
% Table generated by Excel2LaTeX from sheet 'Sheet1'
\begin{table}[htbp!]
  \centering
  
    \begin{tabular}{|c|c|c|c|c|c|}
    \hline
    $n$ & $c(n)$  & $\min \, \overline{\beta}(G)$ & $m(G)$  & $\max \, \overline{\beta}(G)$ & $M(G)$ \\
    \hline
    4     & 1     & 0     & 1     & 0     & 1 \\
    6     & 2     & -2    & 2     & -2    & 2 \\
    8     & 5     & -2     & 1     & 0    & 4 \\
    10    & 14    & -2    & 14    & -2    & 14 \\
    12    & 57    & -2    & 34    & 0     & 23 \\
    14    & 341   & -2    & 341   & -2    & 341 \\
    16    & 2828  & -2    & 2299  & 0     & 529 \\
    18    & 30468 & -2    & 30468 & -2    & 30468 \\
    \hline
    \end{tabular}%
    \caption{Nice decision number for cubic graphs}
  \label{tab:addlabel}%
\end{table}%

% Table generated by Excel2LaTeX from sheet 'Sheet1'
\begin{table}[htbp!]
  \centering
    \begin{tabular}{|c|c|c|c|c|c|}
    \hline
    $n$ & $c(n)$  & $\min \, \gamma(G)$ & $m(G)$  & $\max \, \gamma(G)$ & $M(G)$ \\
    \hline
   4     & 1     & 2     & 1     & 2     & 1 \\
    6     & 2     & 2     & 2     & 2     & 2 \\
    8     & 5     & 4     & 5     & 4     & 5 \\
    10    & 14    & 4     & 8     & 6     & 6 \\
    12    & 57    & 4     & 31    & 8     & 1 \\
    14    & 341   & 6     & 338   & 10    & 1 \\
    16    & 2828  & 6     & 1718  & 8     & 1110 \\
    18    & 30468 & 6     & 8166  & 10    & 121 \\
    \hline
    \end{tabular}%
    \caption{Good decision number for cubic graphs}
  \label{tab:addlabel}%
\end{table}%
% Table generated by Excel2LaTeX from sheet 'Sheet1'
\begin{table}[htbp!]
  \centering
  
    \begin{tabular}{|c|c|c|c|c|c|}
    \hline
    $n$ & $c(n)$  & $\min \, \overline{\gamma}(G)$ & $m(G)$  & $\max \, \overline{\gamma}(G)$ & $M(G)$ \\
    \hline
     4     & 1     & 2     & 1     & 2     & 1 \\
    6     & 2     & 4     & 2     & 4     & 2 \\
    8     & 5     & 4     & 3     & 6     & 2 \\
    10    & 14    & 6     & 13    & 8     & 1 \\
    12    & 57    & 6     & 25    & 8     & 32 \\
    14    & 341   & 8     & 335   & 10    & 6 \\
    16    & 2828  & 8     & 795   & 10    & 2033 \\
    18    & 30468 & 10    & 29692 & 12    & 776 \\
    \hline
    \end{tabular}%
    \caption{Excellent decision number for cubic graphs}
  \label{tab:addlabel}%
\end{table}%

\end{document}